\documentclass{nature}

\usepackage{graphicx}
\usepackage{amssymb,amsfonts,amsmath}
\usepackage{epsfig}          
\usepackage{picinpar}
\usepackage{amsthm}

\newtheorem{Prop}{Proposition}
\newtheorem{Theo}{Theorem}
\newtheorem{Algo}{Algorithm}

\title{Multiplex Structures: Patterns of Complexity in Real-World Networks}

\author{Bo Yang$^1$, Jiming Liu$^2$}

\begin{document}

\maketitle

\begin{affiliations}
\item Jilin University, Changchun city, China. Email: ybo@jlu.edu.cn
 \item Hong Kong Baptist University, Hong Kong. Email: jiming@comp.hkbu.edu.hk
\end{affiliations}



\begin{abstract}
Complex network theory aims to model and analyze complex systems
that consist of multiple and interdependent components. Among all
studies on complex networks, topological structure analysis is of
the most fundamental importance, as it represents a natural route to
understand the dynamics, as well as to synthesize or optimize the
functions, of networks. A broad spectrum of network structural
patterns have been respectively reported in the past decade, such as
communities, multipartites, hubs, authorities, outliers, bow ties,
and others. Here, we show that most individual real-world networks
demonstrate multiplex structures. That is, a multitude of known or
even unknown (hidden) patterns can simultaneously situate in the
same network, and moreover they may be overlapped and nested with
each other to collaboratively form a heterogeneous, nested or
hierarchical organization, in which different connective phenomena
can be observed at different granular levels. In addition, we show
that the multiplex structures hidden in exploratory networks can be
well defined as well as effectively recognized within an unified
framework consisting of a set of proposed concepts, models, and
algorithms. Our findings provide a strong evidence that most
real-world complex systems are driven by a combination of
heterogeneous mechanisms that may collaboratively shape their
ubiquitous multiplex structures as we observe currently. This work
also contributes a mathematical tool for analyzing different sources
of networks from a new perspective of unveiling multiplex
structures, which will be beneficial to multiple disciplines
including sociology, economics and computer science.
\end{abstract}

\section{Introduction}
Complex network analysis provides a novel approach to examining how
networked systems in nature are originated and evolving according to
what basic principles, and moreover armed with such discovered
principles, constructing efficient, robust as well as flexible
man-made networked systems under different constraints. Among all
studies about complex networks, structure analysis is the most
fundamental, and the ability to discover and visualize the
underlying structure of a real-world network in question will be
greatly helpful for both topological and dynamic analysis applied to
it \cite{report2006}. So far, scientists have uncovered a multitude
of structural patterns ubiquitously existing in social, biological,
ecological or technological networks; they may be microscopic such
as motifs \cite{Milo2002}, mesoscopic such as
modularities\cite{Newman2002} or macroscopic such as small world
\cite{Watts1998} and scale-free phenomena \cite{Barabasi1999}. These
structural patterns observed at different granular levels may
collectively unveil the secrets hidden in complex networked systems.
All these works have greatly triggered the common interesting as
well as boosted the progress of exploring complex networks in both
scientific and engineering domains. However, the topological
structure analysis of complex networks, even restricted to the
mesoscopic level, remains one of the major challenges in network
theory mainly because most of the real-world networks are usually
resulted from a combination of heterogeneous mechanisms which may
collaboratively shape their non-trivial structures. More
specifically one can raise the following issues.

Above all, beyond modularity the most extensively studied at the
mesoscopic level \cite{Newman2002,report2010}, a wide spectrum of
structural patterns have been reported in the literature including
bipartites or more generally multipartites
\cite{bipartite2003,bipartite2004,bipartite2009}, hubs , authorities
and outliers \cite{K1999,Albert2000,Sporns2007}, bow ties
\cite{Bowtie1999,bowtie2000,bowtie2003} or others. Moreover, these
miscellaneous patterns may simultaneously coexist in the same
networks, or  they may even overlap with each other such as the
overlaps between communities \cite{Palla2005}. Fig. \ref{lesmis}
shows an illustrative example, in which a social network encoding
the co-appearances of 77 characters in the novel "Les Miserables" is
visualized in terms of both network and matrix representations. One
would observe two hubs and a number of outliers coexisting with four
well-formed communities. The two hubs, corresponding to
\emph{Valjean} and \emph{Javert}, are the main characters of the
novel, and their links are across all other clusters by connecting
about 48\% of the overall characters. It indicates that the two
roles interacting with different characters in different chapters
are the two main clues going through the whole story. Four detected
communities can be seen as four relatively independent social
cliques. As an example, we go into details of group 4 that is almost
separated from the rest of the network. Interestingly, this small
social clique consists of 4 parisian students Tholomyes, Listolier,
Fameuil and Blacheville with their respective lovers Fantine,
Dahlia, Zephine and Favourite. Group 5 consists of 39 outliers,
which connect to either two hubs or one of four communities by only
a few links. Correspondingly, these outliers are the supporting
roles of this novel. Besides this example, more complex structural
patterns in real-world networks will be demonstrated and discussed
in the rest of this article.

Secondly, such multiple structural patterns may be nested. That is,
real-world networks can contain hierarchical organizations with
heterogeneous patterns at different levels. In the literature,
hierarchical structures are usually studied in a homogeneous way, in
which patterns observed in each layer of hierarchies show homophily
in terms of either fractal property \cite{hierarchyscience2004} or
modularity in more general cases \cite{hierarchy2006,hierarchy2007}.
A recent study reveals that the patterns demonstrated in each layer
of a dendrogram can be assortative (or modular) or disassortative
(or bipartite) \cite{hierarchy2008}. The ability to observe
heterogeneous patterns at different levels brings new clues for
understanding the dynamics of real-world complex networks.

Due to the above two reasons, for an exploratory network about which
one often knows little, it is very hard to know what specific
structural patterns can be expected and then be obtained by what
specific tools. Biased results will be caused if an inappropriate
tool is chosen; even if we know something about it beforehand, it is
still difficult for a tool exclusively designed for exploring very
specific patterns, say modularity, to satisfactorily uncover a
multiple of coexisting patterns possibly overlapped or even nested
with each other (we call them \emph{multiplex structures} in this
work) from networks. To the best our knowledge, there have been no
studies in the literature being able to address both of the above
issues. On the other hand, human beings have the capability of
simultaneously discovering multiplex and significant structural
patterns for various objectives. It has been believed that this kind
of capability is an important form of human cognitive and
intelligent functions \cite{multiplexform2008}.

Back to the matrix as shown in the Fig. \ref{lesmis}, one would
observe an intuitive phenomenon: Any significant pattern contained
in the underlying structures of the network can be statistically
highlighted by a group of homogenous individuals with identical or
quite similar connectivity profiles. For instance, individuals from
the same communities prefer to intensively interact with each other
but rarely interact with the rest; hubs would prefer to connect many
individuals from different parts of the whole network, whereas all
outliers tend to seldom play with others by emitting only a few
connections. Based on this naive observation, if one can group the
majority of individuals into reasonable clusters according to their
connectivity profiles, the coexisting structural patterns can be
unveiled by further inferring the linkage among clusters. In this
way, the first issue listed above can be promisingly solved.

The idea of grouping nodes into \emph{equivalent} clusters in terms
of their connection patterns is quite similar with the philosophy of
the \emph{blockmodeling} first proposed by Lorrain and White
\cite{block1971}, in which nodes with \emph{structural equivalence}
(defined in terms of local neighborhood configurations) or more
generally \emph{regular equivalence} \cite{block1983a} (defined in
terms of globally physical connections to all others) or more softly
\emph{stochastic equivalence} \cite{block1981,block1983b} (defined
in terms of linking probabilities between groups) will be grouped
into the same blocks.

Based on the same idea, a very related study has been proposed
recently by Newman and Leicht, which first (to our knowledge) shows
the motivation to detect unpredefined structural patterns from
exploratory networks \cite{pans2007}. From the perspective of
machine leaning, their method can be seen as a version of naive
Bayes algorithm applied to networks, whose objective is to group
nodes with similar connection features into a predefined number of
clusters. Although their work only shows the ability to determine
whether an exploratory network is assortative or disassortative by
manually analyzing the obtained clusters, it has provided one good
proof supporting that the idea of grouping nodes into equivalent
clusters can be an initial step of the whole process aiming to
unveil multiplex structures from networks.

In this work, we will propose a novel model by introducing the
concept of granularity into stochastic connection profiles in order
to properly model multiplex structures, and then show how the task
of recognizing multiplex structures can be reduced to a simplified
version of the isomorphism subgraph matching problem. To test our
ideas and strategies proposed here, different sources of networks
have been analyzed. It is encouraging that our methods show a good
performance, capable of uncovering multiplex structures from the
tested networks in a fully automatic way.

\section{The Model}\label{sec:2}
\subsection{Granular couplings}
We will model the connection profiles of nodes in terms of
probabilities instead of physical connectivity. In this way, it is
expected not only to find out multiplex structures but also to
provide an explicit probabilistic interpretation for these findings
within the Bayesian framework. We term such probabilities as
\emph{couplings} in that they are not just the mathematical measures
subjectively defined for modeling or computing, but they do exist in
many real-world systems, encoding different physical meanings such
as social preferences in societies, predation habits in ecosystems,
co-expression regularities in gene networks or co-occurrence
likelihood of words in languages, which will be valuable to predict
their situated systems.


Here the notation of granularity should be interpreted in terms of
the resolution and precision. On one hand, in our model we define
two kinds of couplings with different resolutions, i.e., \emph{node
couplings} and \emph{block couplings}. Formally, we define
\emph{node feedforward-coupling} matrix $P_{n\times n}$ and
\emph{node feedback-coupling} matrix $Q_{n\times n}$, where $p_{ij}$
and $q_{ij}$ respectively denote the probabilities that node $i$
expects to couple with or to be coupled by node $j$. In the cases of
indirected networks we have $P=Q$ (see SI for proofs). We assume
nodes will independently couple with others regulated by such
couplings. Nodes with similar feedforward- as well as feedback-
coupling distributions will be clustered into the same blocks. In
terms of matrices, homogeneous feedforward- and feedback-couplings
guarantee homogeneous row and column connection profiles,
respectively. Correspondingly, given the block number $K$, we define
\emph{block feedforward-coupling} matrix $\Phi_{K\times K}$ and
\emph{block feedback-coupling matrix} $\Psi_{K\times K}$, where
$\phi_{pq}$ and $\psi_{pq}$ respectively denote the probabilities
that block $C_p$ expects to couple with or to be coupled by block
$C_q$. We will show later that block couplings can be inferred from
node couplings and vice versa. Node couplings with a fine
granularity are used to model networks in order to capture their
local information as much as possible; while block couplings with a
coarser granularity are used to define and recognize global
structural patterns by intentionally neglecting trivial details. On
the other hand, in the nested patterns, node couplings and block
couplings in different hierarchies will have different precisions in
order to properly abstract and construct hierarchical organizations.
In our model, the couplings on higher layers are the approximations
of the related ones on the lower layers. Therefore, as the layers
moving from the bottom (corresponding to the original networks) to
the top (corresponding to the finally reduced networks) of the
nested organizations, the precision of node or block couplings will
gradually degenerate.


\subsection{Defining multiplex structures}
The main steps of our strategies for detecting multiplex structures
from networks can be stated as follows: 1) simultaneously estimating
all kinds of couplings mentioned above and clustering all nodes into
nested blocks with a proposed granular blocking algorithm; and 2) in
each layer of the nested blocks, recognizing structural patterns by
matching predefined isomorphism subgraphs from a reduced blocking
model in which trivial couplings are neglected, as illustrated in
Fig. \ref{lesmis}d. Multiplex structures can be defined in terms of
blocks and their couplings. Fig. \ref{blocking} shows a schematic
illustration by means of some conceptual networks. By referring to
them, we give following definitions.

A \emph{community} is a self-coupled block. An \emph{authority} is a
self-coupled block which is coupled by a number of other blocks. A
\emph{hub} is a self-coupled block which couples with a number of
other blocks. An \emph{outlier} is a block without self-coupling
which is coupled by a hub or couples with an authority. \emph{A
bow-tie} is a subgraph consisting of a block $b$ and two sets of
blocks $S_l(b)$ and $S_r(b)$, which satisfy with: 1) $b$ is coupled
by and couples with the blocks of $S_l(p)$ and $S_r(p)$,
respectively; 2) the intersection of $S_l(p)$ and $S_r(p)$ is empty
or $\{b\}$; and 3) there are no couplings between $S_l(p)$ and
$S_r(p)$; A \emph{ multipartite} is a subgraph consisting of a set
of blocks without self-couplings which reciprocally couple with each
other. As a special case of multipartite, a \emph{bipartite} is a
subgraph consisting of two blocks without self-loop couplings which
unilaterally or bilaterally coupled with each other. A
\emph{hierarchical organization} is a set of nested blocks, in which
block subgraphs in lower layers are directly or indirectly nested in
the bigger blocks on higher layers.

The above definitions imply that there may exist overlaps between
different patterns in the sense that the same blocks can be
simultaneously involved in a multitude of subgraphs. For example, a
block which is determined as a community can be also a hub, a
authority or the core of one bow tie. Moreover, beyond the
predefined patterns, users are allowed to define novel even more
complex patterns by designing new subgraphs of blocks, which can be
identified by matching their isomorphism counterparts from blocking
models.

\subsection{Granular blocking model}

Let $N=(V,E)$ be a directed and binary network, where $V(N)$ denotes
the set of nodes and $E(N)$ denotes the set of directed links. In
the case of undirected network, we suppose there are two direct
links between each pair of nodes. Let $A_{n\times n}$ be the
adjacency matrix of $N$, where $n$ is the number of nodes.

Suppose all nodes of $N$ are divided into $L(1\leq L \leq n)$
blocks, denotes by $B_{n\times L}$, where $b_{il}=1$ if node $i$ is
in the block $l$, otherwise $b_{il}=0$. When each block is
considered to be inseparable, the \emph{granularity} of network $N$
can be measured by the average size of blocks $g=n/L$. As $g$
increasing from 1 to $n$, the granularity of $N$ degenerates from
the finest to the coarsest. Let $B_g$ denote the blocking model with
a granularity $g$, and we expect to cluster all its blocks into a
reasonable number of clusters so that the nodes of blocks within the
same cluster will demonstrate homogeneous coupling distributions.
Let matrix $Z_{L\times K}(1\leq K\leq L)$ denote such clusters,
where $K$ is the cluster number and $z_{lk}=1$ if block $l$ is
labeled by cluster $k$, otherwise $z_{lk}=0$. Since the coupling
distributions of nodes within the same clusters are expected to be
homogeneous, one can characterize such distributions for each
cluster instead of for each node. Given $Z$, define $\Theta_{K\times
n}$, where $\theta_{kj}$ denotes the probability that any node out
of cluster $k$ expects to couple with node $v_j$; define
$\Delta_{K\times n}$, where $\delta_{kj}$ denotes the probability
that any node out of cluster $k$ expects to be coupled by node
$v_j$; define $\Omega=(\omega_1,\cdots,\omega_K)^T$, where
$\omega_k$ denotes the prior probability that a randomly selected
node will belong to the cluster $C_k$. It is easy to show
$P=B_gZ\Theta$ and $Q=B_gZ\Delta$ (see SI).

Let $X=(K,Z,\Theta,\Delta,\Omega)$ be a model with respect to $N$
and $B_g$. We expect to select an optimal $X$ from its hypothesis
space to properly fit as well as to precisely predict the behaviors
of $N$ under $B_g$ in terms of node couplings characterized by it.
According to the MAP principle (maximum a posteriori), the optimal
$X$ for a given network $N$ under $B_g$ will be the one with the
maximum posterior probability. Moreover, we have: $P(X|N,B_g)\propto
P(N|X,B_g)P(X|B_g)$, where $P(X|N,B_g)$, $P(N|X,B_g)$ and $P(X|B_g)$
denote the posteriori of $X$ given $N$ and $B_g$, the likelihood of
$N$ given $X$ and $B_g$ and the priori of $X$ given $B_g$,
respectively.

\section{Learning Methods}
\subsection{Likelihood maximization} We first consider the simplest
case by assuming all the prioris of $X$ given $B_g$ are equal. In
this case, we have: $P(X|N,B_g)\propto P(N|X,B_g)$. Let
$L(N|X,B_g)=\ln P(N|X,B_g)$, we have (see SI):

\begin{equation}
L(N|X,B_g)=\sum_{l=1}^L\sum_{b_{il}\neq 0}\ln \sum_{k=1}^K
\Pi_{j=1}^n
f(\theta_{kj},a_{ij})f(\delta_{kj},a_{ji})\omega_k\label{likelihood}
\end{equation}

\noindent where $f(x,y)=x^y(1-x)^{1-y}$.

Let $L(N,Z|X,B_g)$ be the log-likelihood of the joint distribution
of $N$ and $Z$ given $X$ and $B_g$, we have (see SI):

\begin{eqnarray}
L(N,Z|X,B_g)= \sum_{l=1}^L\sum_{b_{il}\neq 0}\sum_{k=1}^K
z_{lk}(\sum_{j=1}^n (\ln f(\theta_{kj},a_{ij})+ \ln
f(\delta_{kj},a_{ji}))+\ln\omega_k)\label{joint}
\end{eqnarray}

Considering the expectation of $L(N,Z|X,B_g)$ on $Z$, we have:

\begin{eqnarray}
E[L(N,Z|X,B_g)]= \sum_{l=1}^L\sum_{b_{il}\neq 0}\sum_{k=1}^K
\gamma_{lk}(\sum_{j=1}^n (\ln f(\theta_{kj},a_{ij})+ \ln
f(\delta_{kj},a_{ji}))+\ln\omega_k)\label{Ejoint}
\end{eqnarray}

\noindent where $E[z_{lk}]=\gamma_{lk}=P(y=k|b=l,X,B_g)$, i.e., the
probability that block $l$ will be labeled as cluster $k$ given $X$
and $B_g$.

Let $J=E[L(N,Z|X,B_g)]+\lambda(\sum_{k=1}^K \omega_k-1)$ be a
Lagrangian function constructed for maximizing $E[L(N,Z|X,B_g)]$
with a constraint $\sum_{k=1}^K \omega_k=1$, we have:
\begin{equation}
\left\{\begin{array}{ll}
\frac{\partial J}{\partial \theta_{kj}}=0 \\
\frac{\partial J}{\partial \delta_{kj}}=0 \\
\frac{\partial J}{\partial \omega_k}=0 \\
\frac{\partial J}{\partial \lambda}=0
\end{array} \right. \Rightarrow \left\{ \begin{array}{l}
\theta_{kj}=\frac{\sum_{l=1}^L\sum_{b_{il}\neq 0}
a_{ij}\gamma_{lk}}{\sum_{l=1}^L\sum_{b_{il}\neq
0}\gamma_{lk}}\\
\delta_{kj}=\frac{\sum_{l=1}^L\sum_{b_{il}\neq
0} a_{ji}\gamma_{lk}}{\sum_{l=1}^L\sum_{b_{il}\neq 0}\gamma_{lk}} \\
\omega_k=\frac{\sum_{l=1}^L\sum_{b_{il}\neq 0}\gamma_{bk}}{n}
\end{array} \right.\label{eq:theta}
\end{equation}

According to the Bayesian theorem,we have (see SI):
\begin{equation}
\gamma_{lk}=\frac{1}{\sum_{i=1}^n b_{il}}\sum_{b_{il}\neq
0}\frac{\Pi_{j=1}^n f(\theta_{kj},a_{ij})f(\delta_{kj},a_{ji})
\omega_k}{\sum_{k=1}^K \Pi_{j=1}^n
f(\theta_{kj},a_{ij})f(\delta_{kj},a_{ji}) \omega_k}
\label{eq:gamma}
\end{equation}

Using the similar treatment as proposed by Dempster and Laird
\cite{Dempster1977}, we can prove that a local optimum of
Eq.\ref{likelihood} will be guaranteed by recursively calculating
Eqs.\ref {eq:theta} and \ref{eq:gamma} (see SI). The time complexity
of this iterative computing process is $O(In^2K)$, where $I$ is the
iterations required for convergence, which is usually quite small.
An approximate but much faster version with a time $O(ILnK)$ is
given in the SI.


\subsection{Priori approximation} Without considering the priori of
the model, the above-proposed likelihood maximization algorithm will
be cursed by the overfitting problem. That is, $L(N|X,B_g)$ will
monotonically increase as $K$ approaching to $L$. In this section,
we will discuss how to approximately estimate the prior $P(X|B_g)$
by means of the information theory.

Note that $1\leq K\leq L=n/g$, which implies that the coarser
granularity the smaller $K$. It will be shown in the following that
a smaller $K$ will indicate a less complexity of $X$. So, we have: a
coarser granularity prefer simpler models, which can be
mathematically written as $P(X|B_g)= \eta(X)^g$, where the function
$\eta(X)$ measures the complexity of $X$ in terms of its parameters.
In this work, we set $\eta(X)=P(X|B_1)=P(X)$, where $P(X)$ is the
priori of $X$ in the hypothesis space in which $K$ can be freely
valued from 1 to $n$. According to Shannon and Weaver
\cite{shannon1949}, $\ln(1/P(X))$ is the minimum description length
of $X$ with a prior $P(X)$ in its model space. Let $OC(X)$ denote
the optimal coding schema for $X$, and let $L_{OC}(X)$ be the
minimum description length of $X$ under this schema. We have: $-\ln
P(X|B_g)=-g\ln P(X)=gL_{OC}(X)$. Now, to estimate the prior
$P(X|B_g)$ is to design a good coding (or compressing) schema as
close to $OC$ as possible.

In terms of the parameter of $X$, i.e.,$\Theta$, $\Delta$, $Z$ and
$\Omega$, we have (see SI):

\begin{equation}
\Phi=\Theta B_gZ D^{-1}, \quad \Psi=\Delta B_gZ D^{-1}
\end{equation} \noindent where $D=diag(n\Omega)$.

Note that matrices $Z$ and $B_g$ can be compressed into a map $y$,
where $y(i)=k$ if the entry $(i,k)$ of $B_gZ$ is equal to one. Given
$y$, $\Phi$ and $\Psi$, node couplings $p_{ij}$ and $q_{ij}$ can be
measured by:

\begin{equation}
p_{ij}= \phi_{y(i),y(j)},\quad q_{ij}= \psi_{y(i),y(j)} \label{MDL}
\end{equation}

Eq.\ref{MDL} says that, all node couplings can be approximately
characterized by $y$, $\Phi$ and $\Psi$. In other word, the
compressing schema close to $OC(X)$ we have found out is:
\begin{equation}
\widehat{OC}(X) = (K, \Phi_{K\times K}, \Psi_{K\times K}, y_{n\times
2})
\end{equation}

Now, we have,

$L_{\widehat{OC}}(X)=1\times (-\ln\frac{1}{1})+
2K^2(-\ln\frac{1}{K^2})+2n(-\ln\frac{1}{2n})$

$\quad\quad\quad\quad= 2K^2\ln K^2 + 2n\ln2n$

Moreover, we have:
\begin{equation}
\begin{array}{l}
\arg\max_{X} P(X|N,B_g) \\
=\arg\min_{X} (-\ln P(N|X,B_g) -\ln P(X|B_g)) \\
=\arg\min_{X} (-\ln P(N|X,B_g) + gL_{\widehat{OC}}(X)) \\
=\arg\min_{X} (-L(N|X,B_g) + g(2K^2\ln K^2 + 2n\ln2n)) \\
=\arg\min_{X} (-L(N|X,B_g) + 2gK^2\ln K^2)
\end{array} \label{objective}
\end{equation}

Eq.\ref{objective} tries to seek a good tradeoff between the
accuracy of model (or the precision of fitting observed data)
measured by the likelihood of a network, and the complexity of a
model (or the generalization ability to predict new data) measured
by its optimal coding length.

\subsection{Model selection} For a given $K$, the penalty term is a
constant, and thus to maximize $P(X|N,B_g)$ is to maximize
$L(N|X,B_g)$. In our algorithm, $K$ will be systematically checked
from 1 to $L$, and the model with the minimum sum of negative
likelihood and penalty will be returned as the optimal one. In the
landscape of $K$ and $P(X|N,B_g)$, a well-like curve will be shaped
during the whole search process (see SI). So, in practice, rather
than mechanically checking $K$ for exact $L$ times, an ongoing
searching can be stopped after it has safely passed the well bottom.
By means of this greedy strategy, the efficiency of our algorithm
would be greatly improved. The complete algorithm is given in the
SI.

\subsection{Hierarchy construction}  Assume we have constructed $h$
layers, in which the $i-th$ layer corresponds to a blocking model
characterized by $B_{g_i}$. Now, we want to construct the $(h+1)-th$
layer by selecting an $X$ with a maximum posterior given a set of
blocking models on different layers. We have shown that (see SI):
$P(X|N,B_{g_1},\cdots,B_{g_h})\propto
P(N|X,B_{g_h})P(X|B_{g_h})\propto P(N|X,B_{g_h})P(X)^{g_h}$. This
Markov-like process indicates that the new model to be selected for
layer $h$ is only based on the information of the layer $h-1$. So,
for a given network, its hierarchical organization can be
incrementally constructed as follows.

Firstly, we construct the first layer by taking each node as one
block, and cluster it into $K_1$ clusters by selecting an model with
a maximum $P(X|N,B_1)$. Next, we form $B_{n/K_1}$ by capsuling each
cluster on the first layer as one block, and cluster these blocks
into $K_2$ clusters by selecting an new model with a maximum
$P(X|N,B_{n/K_1})$, which forms the second layer. We repeat the same
procedure until it converges (the cluster number obtained keeps
fixed). In this way, the number of layers of a hierarchical
organization will be automatically determined. The above procedure
can be seen as a semi-supervised learning process; in the cases that
the granularity is more than one, we have already known a priori
that which nodes will be definitely within the same clusters. As the
layer in the constructed hierarchy increases, the homogeneity of the
nodes within the same cluster in terms of their connection profiles
keeps degenerating, and correspondingly more global patterns are
allowed to be observed by tolerating such increasing diversities.

\subsection{Isomorphism subgraph matching} Based on the obtained
blocks in the level $h+1$ of the hierarchy, all potential patterns
hidden in the level $h$ can be revealed through an isomorphism
subgraph matching procedure, whose input is the block
feedforward-coupling matrix $\Phi$. First we construct a reduced
blocking model by taking each block as one node, and for each pair
of blocks $p$ and $q$ we insert a link from $p$ to $q$ if
$\phi_{pq}$ is above a threshold computed based on $\Phi$ (see SI).
Then for each block, we pick up the matched isomorphism subgraphs it
will be involved in, and put them into categorized reservoirs
labeled by different patterns. During this procedure, the subgraph
to be put into a reservoir will be discarded if it is a subgraph of
an existing one, as illustrated by Fig.\ref{worldtrade}d.

\section{Applications}
\subsection{Exploring the cash flow patterns of the world trade system}
The discovered multiplex structures as well as their granular
couplings can be used to understand some dynamics of networks. Here
we give one example to show how cash possibly flow through a world
trade net. Fig. \ref{worldtrade}a shows a directed network encoding
the trade relationship among eighty countries in the world in 1994,
which was originally constructed by Nooy \cite{nooy} based on the
commodity trade statistics published by the United Nations. Nodes
denote countries, and each arc $i\rightarrow j$ denotes the country
$i$ imported high technology products or heavy manufactures from the
country $j$. Analogous to the "structural world system positions"
initially suggested by Smith and White based on their analysis of
the international trade from 1965 to 1980 \cite{GDP}, the eight
countries in 1994 were categorized into three classes according to
their economic roles in the world: core, semi-periphery and
periphery \cite{nooy}. Accordingly, in the visualized network, the
countries labeled by them are distributed along three circles from
inside to outside, respectively.

A two-layer hierarchical organization has been constructed by our
system, as illustrated in Fig. \ref{worldtrade}d. A reduced blocking
model is shown in the first layer by neglecting trivial couplings,
in which six blocks are obtained. By referring the matrix of the
network as presented in Fig. \ref{worldtrade}b, one can observe that
the nodes within the same blocks demonstrate homogeneous row as well
as column connection profiles. By referring to their couplings, ten
isomorphism subgraphs of the patterns as defined in Fig.
\ref{blocking} are recognized, which, respectively, are one
authority, four communities, three hubs and two outliers, as shown
in Fig. \ref{worldtrade}e. Quite a few interesting things can be
read from these uncovered multiplex structures. Globally, the
countries near center tend to have larger out-degrees, while those
far from center have smaller even zero out-degrees. Specifically,
(1) according to the coupling strength from strong to weak, three
detected hubs can be ranked into the sequence of blocks 4, 1 and 3.
The first two hubs consist of the "core" of the trade system except
for Spain and Denmark; (2) the countries from blocks 3, 2 and 6
consist of the backbone of the "semi-periphery"; (3) more than a
half of "periphery" countries (10 of 17) have zero out-degrees; (4)
interestingly, the detected blocks are also geography-related, as
illustrated in Fig. \ref{worldtrade}c. Most countries of hubs 4 and
1 are from western Europe expect America, China and Japan; most of
hub 3 are from southeastern Asia; most of the community block 6 are
from north or south America; most of the outlier block 2 are from
eastern Europe; most of the outlier block 5 are from Africa or some
areas of Asia.

In the second layer, a macroscopic hub-outlier pattern with strong
couplings is recognized. Hub blocks 4 and 1 in the first layer
collectively form a more global hub of the whole network as the core
of the entire trade system; other blocks form a more global outlier
of the network corresponding to the semi-periphery and periphery of
the trade system. This nested hub-outlier patterns perhaps give us
an evidence about how cash flowed through the world in different
levels in 1994. Note that arc $i\rightarrow j$ denotes that country
$i$ imported commodities from country $j$, which also indicates that
the spent cash has flowed from $i$ to $j$. In this way, one can
image cash flows along these arcs from one country to another.
According to the global pattern in the second layer, the dominant
cash flux will flow from the core countries to themselves with a
probability 0.6, and to the rest with a probability 0.57. Locally,
the blocking model in the first layer shows the backbone of the cash
flow through the entire world with their respective strength in
terms of block couplings, as illustrated in Fig. \ref{worldtrade}d.

\subsection{Mining granular association rules from networks}
When a network encodes the co-occurrence of events, its underlying
node- or block-coupling matrices would imply the probabilistic
associations among these events in different granular levels,
respectively. Formally, we have: \emph{node association rule (NAR)}:
$i\rightarrow j \:\langle p_{ij}\rangle$, and \emph{block
association rule (BAR)}: $B_p\rightarrow
B_q\:\langle\phi_{pq}\rangle$. A \emph{NAR}  means that event $j$
would happen with a probability $p_{ij}$ given event $i$ happens. A
\emph{BAR} means that any event of block $q$ would happen with a
probability $\phi_{pq}$ given any event of block $p$ happens. Such
association rules with different granularities can be used in making
prediction in a wide range of applications, such as online
recommender systems. We will demonstrate this idea with a political
book co-purchasing network compiled by V. Krebs, as given in Fig.
\ref{polbooks}(a), where nodes represent books about US politics
sold by the online bookseller Amazon.com, and edges connects pairs
of books that are frequently co-purchased, as indicated by the
"customers who bought this book also bought these other books"
feature on Amazon. Moreover, these books have been labeled as
"liberal", "neutral" or "conservative" according to their stated or
apparent political alignments based on the descriptions and reviews
of the books posted on Amazon.com \cite{polbook}.

A two-layer hierarchical organization has been detected by our
system as shown in Fig. \ref{polbooks}. The blocking model of the
first layer is shown in Fig. \ref{polbooks}(b). By matching
isomorphism subgraphs in its reduced blocking model, nine patterns
are recognized, which respectively are five communities (blocks
1,2,4,6,7), two cores (blocks 2 and 7), two outliers (blocks 3 and
5) and a bow tie (blocks 1,2,3). Note that, in indirected networks,
the core of a bow tie (block 2 in this case) can be seen as the
overlapping part of its two wings (blocks 1 and 3 in this case) by
neglecting the direction of links. In the second layer, a
macroscopic community structure is recognized, as shown in Fig.
\ref{polbooks}(c). Interestingly, the left and right communities can
be globally labeled as "left-wing" and "right-wing" according to the
types of the books they contain respectively. Such a global pattern
can be seen as one macroscopic \emph{BAR}: the books with common
labels would be co-purchased with a great chance (about 15\%);
while, those with different labels are rarely co-purchased (only
with a chance of 1\%).

When zooming in to both global communities in the second layer, one
will obtain $7 \times 7$ mesoscopic \emph{BAR}s encoded by the
block-coupling matrix $\Phi$, as illustrated by the weighted arrow
lines Fig. \ref{polbooks}(b). As an example, we list the \emph{BRA}s
related to the block 2 in a decreasing sequence of association
strength. $B_2\rightarrow B_2\langle0.60\rangle$; $B_2\rightarrow
B_3\langle0.44\rangle$; $B_2\rightarrow B_1\langle0.21\rangle$;
$B_2\rightarrow B_4\langle0.06\rangle$; $B_2\rightarrow
B_6\langle0.04\rangle$; $B_2\rightarrow B_5\langle0\rangle$;
$B_2\rightarrow B_7\langle0\rangle$. Such mesoscopic association
rules would help booksellers adaptively adjust their selling
strategies to determine what kinds of stocks they should increase or
decrease based on the statistics of past sales. For example, if they
find the books labeled as block 2 are sold well, they may
correspondingly increase the order of books labeled as blocks 1 and
3 besides block 2, while they may simultaneously decrease the order
of books labeled as blocks 5 or 7.

More specifically, with the aid of \emph{NAR}s, booksellers would be
able to estimate the chance that customers will spend their money on
a book $j$ if they have already bought book $i$ by referring to
$i\rightarrow j \langle p_{ij}\rangle$. Such microscopic rules would
provide booksellers the suggestions on what specific books should be
listed according to what sequence in the recommending area of the
web page advertising a book. For example, for the book \emph{"The
Price of Loyalty" }, the most worth recommended books are listed as
follows according to the coupling strength to it: \emph{Big
Lies}$\langle0.91\rangle$; \emph{Bushwhacked}$\langle0.73\rangle$;
\emph{Plan of Attack}$\langle0.73\rangle$; \emph{The Politics of
Truth}$\langle0.73\rangle$; \emph{The Lies of George
W.Bush}$\langle0.73\rangle$; \emph{American
Dynasty}$\langle0.64\rangle$; \emph{Bushwomen}$\langle0.64\rangle$;
\emph{The Great Unraveling}$\langle0.64\rangle$; \emph{Worse Than
Watergate}$\langle0.64\rangle$.

\section{Conclusions}\label{sec:4}
In this work, we have shown through examples that the structural
patterns coexisting in the same real-world complex network can be
miscellaneous, overlapped or nested, which collaboratively shape a
heterogeneous hierarchical organization. We have proposed a
framework based on the concept of granular couplings and the
proposed granular blocking model to uncover such multiplex
structures from networks. From the output of patterns, hierarchies
and granular couplings generated by our approach, one can analyze or
even predict some dynamics of networks, which are helpful for both
theoretical studies and practical applications.

Moreover, based on the rationale behind this work, we suggest that
the evolution of a real-world network may be driven by the
co-evolution of its structural patterns and its underlying
couplings. On one hand, statistically significant patterns would be
gradually highlighted and emergently shaped by the aggregations of
homogeneous individuals in terms of their couplings. On the other
hand, individuals would adaptively adjust their respective couplings
according to the currently evolved structural patterns.

\newpage
\centerline{\textbf{Support Information}}

\appendix{}{}

\section{Proofs and algorithms}\label{sec:1}

\begin{Prop}\label{prop:1}
For an indirected network, its feedforward-coupling matrix $P$ is
equal to its feedback-coupling matrix $Q$.
\end{Prop}

\begin{proof}

\noindent\smallskip

\noindent $p_{ij}=P(i\rightarrow j|y=k)$

\noindent $=\sum_{k=1}^K m_{ik}P(i\rightarrow j|y=k)$

\noindent$=\sum_{k=1}^K (\sum_{l=1}^L b_{il}z_{lk})P(i\rightarrow
j|y=k)$

\noindent$= \sum_{l=1}^L \sum_{k=1}^K b_{il}z_{lk}\theta_{kj}$

\noindent where $i\rightarrow j$ denote the event that node $v_i$
couples with $v_j$, and $y=k$ denote the event that $v_i$ is labeled
by cluster $k$; $m_{ik}=1$ if $v_i$ is labeled by cluster $k$,
otherwise $m_{ik}=0$. So we have

\centerline{$P=B_gZ\Theta$.}

\noindent $q_{ij}=P(i\dashleftarrow j|y=k)$

\noindent $=\sum_{k=1}^K m_{ik}P(i\dashleftarrow j|y=k)$

\noindent$=\sum_{k=1}^K (\sum_{l=1}^L b_{il}z_{lk})P(i\dashleftarrow
j|y=k)$

\noindent$= \sum_{l=1}^L \sum_{k=1}^K b_{il}z_{lk}\delta_{kj}$

\noindent where $i\dashleftarrow j$ denote the event that node $v_i$
except to be coupled by $v_j$. So we have

\centerline{$Q=B_gZ\Delta$.}

If $A$ is symmetry, from the Eq.4 in the article, we have

\smallskip
\centerline{$\theta_{kj}=\frac{\sum_{l=1}^L\sum_{b_{il}\neq 0}
a_{ij}\gamma_{lk}}{\sum_{l=1}^L\sum_{b_{il}\neq 0}\gamma_{lk}}=
\frac{\sum_{l=1}^L\sum_{b_{il}\neq 0}
a_{ji}\gamma_{lk}}{\sum_{l=1}^L\sum_{b_{il}\neq
0}\gamma_{lk}}=\delta_{kj}$.}

So we have $P=Q$.
\end{proof}

\bigskip
\begin{Prop}\label{prop:1}
\begin{equation}
L(N|X,B_g)=\sum_{l=1}^L\sum_{b_{il}\neq 0}\ln \sum_{k=1}^K
\Pi_{j=1}^n
f(\theta_{kj},a_{ij})f(\delta_{kj},a_{ji})\omega_k\label{likelihood}
\end{equation}
\end{Prop}

\noindent where $f(x,y)=x^y(1-x)^{1-y}$.

\begin{proof}
Let $v=i$ denote the event that a node with linkage structure
$<a_{i1},\cdots,a_{in},a_{1i},\cdots,a_{ni}>$ will be observed in
network $N$. Let $y=k$ denote the event that the cluster label
assigned to a node is equal to $k$. Let $i\rightarrow_{a_{ij}} j$
denote the event that node $v_i$ link to node $v_j$ or not depending
on $a_{ij}$. Let $i\leftarrow_{a_{ji}} j$ denote the event that node
$v_i$ will be linked by node $v_j$ or not depending on $a_{ji}$. We
have:

\smallskip

\noindent $L(N|X,B_g)=\ln\Pi_{i=1}^nP(v=i)=\sum_{i=1}^n\ln P(v=i)$

\noindent$=\sum_{i=1}^n\ln P((v=i) \cap(\cup_{k=1}^Ky=k))$

\noindent$=\sum_{i=1}^n\ln \sum_{k=1}^K P(v=i,y=k)$

\noindent$=\sum_{i=1}^n\ln \sum_{k=1}^K (P(v=i|y=k)P(y=k))$

\noindent$=\sum_{i=1}^n\ln \sum_{k=1}^K
(P(<a_{i1},\cdots,a_{in},a_{1i},\cdots,a_{ni}>|y=k)P(y=k))$

\noindent$=\sum_{i=1}^n\ln \sum_{k=1}^K (\Pi_{j=1}^n
P(i\rightarrow_{a_{ij}} j|y=k)P(i\leftarrow_{a_{ji}} j|y=k)P(y=k))$

\noindent$=\sum_{i=1}^n\ln \sum_{k=1}^K (\Pi_{j=1}^n
(\theta_{kj}^{a_{ij}}(1-\theta_{kj})^{1-a_{ij}})(\delta_{kj}^{a_{ji}}(1-\delta_{kj})^{1-a_{ji}})\omega_k)$

\noindent$=\sum_{i=1}^n\ln \sum_{k=1}^K (\Pi_{j=1}^n
f(\theta_{kj},a_{ij})f(\delta_{kj},a_{ji})\omega_k)$

\noindent$=\sum_{l=1}^L\sum_{b_{il}\neq 0}\ln \sum_{k=1}^K
(\Pi_{j=1}^n f(\theta_{kj},a_{ij})f(\delta_{kj},a_{ji})\omega_k)$
\end{proof}

\bigskip
\begin{Prop}\label{prop:2}
\begin{eqnarray}
L(N,Z|X,B_g)= \sum_{l=1}^L\sum_{b_{il}\neq 0}\sum_{k=1}^K
z_{lk}(\sum_{j=1}^n (\ln f(\theta_{kj},a_{ij})+ \ln
f(\delta_{kj},a_{ji}))+\ln\omega_k)\label{joint}
\end{eqnarray}
\end{Prop}

\begin{proof} Let $y(i)$ denote the cluster label assigned to node $i$ under the
given partition $Z$, we have:

\noindent $L(N,Z|X,B_g)$

\noindent $=\ln\Pi_{i=1}^n P(v=i,y=y(i))$

\noindent $=\sum_{i=1}^n \ln \sum_{k=1}^K m_{ik}P(v=i,y=k)$

\noindent $=\sum_{i=1}^n\ln \sum_{k=1}^K m_{ik}P(v=i|y=k)P(y=k)$

\noindent $=\sum_{i=1}^n\sum_{k=1}^K \ln(P(v=i|y=k)P(y=k))^{m_{ik}}$

\noindent $=\sum_{i=1}^n\sum_{k=1}^K m_{ik}\ln(P(v=i|y=k)P(y=k))$

\noindent $=\sum_{i=1}^n \sum_{k=1}^K m_{ik}\ln(\Pi_{j=1}^n
(\theta_{kj}^{a_{ij}}(1-\theta_{kj}^{1-a_{ij}})\delta_{kj}^{a_{ji}}(1-\delta_{kj}^{1-a_{ji}}))\omega_k$


\noindent $=\sum_{i=1}^n\sum_{k=1}^K m_{ik}(\sum_{j=1}^n(\ln
f(\theta_{kj},a_{ij})+ \ln f(\delta_{kj},a_{ji}))+\ln\omega_k)$

\noindent $=\sum_{b_{i1\neq 0}}\sum_{k=1}^K z_{1k}(\sum_{j=1}^n(\ln
f(\theta_{kj},a_{ij})+ \ln
f(\delta_{kj},a_{ji}))+\ln\omega_k)+\cdots$

$+\sum_{b_{iL\neq 0}}\sum_{k=1}^K z_{Lk}(\sum_{j=1}^n(\ln
f(\theta_{kj},a_{ij})+ \ln f(\delta_{kj},a_{ji}))+\ln\omega_k)$

\noindent $=\sum_{l=1}^L\sum_{b_{il}\neq 0}\sum_{k=1}^K
z_{lk}(\sum_{j=1}^n(\ln f(\theta_{kj},a_{ij})+ \ln
f(\delta_{kj},a_{ji}))+\ln\omega_k)$.

\end{proof}

Notice that, in the proofs of Props \ref{prop:1} and \ref{prop:2},
all probabilities such as $P(y=k|v=i)$ and $P(y=k)$ are discussed
under the conditions of $X$ and $B_g$. To simplify the equations, we
omit them without losing correctness.
\bigskip

\begin{Prop}
\begin{equation}
\gamma_{lk}=\frac{1}{\sum_{i=1}^n b_{il}}\sum_{b_{il}\neq
0}\frac{\Pi_{j=1}^n f(\theta_{kj},a_{ij})f(\delta_{kj},a_{ji})
\omega_k}{\sum_{k=1}^K \Pi_{j=1}^n
f(\theta_{kj},a_{ij})f(\delta_{kj},a_{ji}) \omega_k}
\label{eq:gamma}
\end{equation}
\end{Prop}

\begin{proof}
let $P(y=k|v=i)$ be the probability that node $i$ belongs to cluster
$k$ given $X$ and $B_g$. We have:

\centerline{ $\gamma_{lk}=P(y=k|b=l,X,B_g)= \sum_{b_{il}\neq
0}\frac{1}{\sum_{i=1}^n b_{il}}P(y=k|v=i)$}

\noindent where $\frac{1}{\sum_{i=1}^n b_{il}}$ is the probability
of selecting node $i$ from block $l$.

\smallskip According to the Bayesian theorem, we have:

\smallskip
\centerline{ $P(y=k|v=i)=\frac{P(y=k)P(v=i|y=k)}{\sum_{k=1}^K
P(y=k)P(v=i|y=k)}$.}

\smallskip
\smallskip Based on the proof of Prop.\ref{prop:1}, we have:

\smallskip
\centerline{ $P(y=k)P(v=i|y=k)=\Pi_{j=1}^n
f(\theta_{kj},a_{ij})f(\delta_{kj},a_{ji}) \omega_k$.}

\smallskip So, we have Eq.\ref{eq:gamma}.
\end{proof}

\bigskip

As an approximate version of Eq.\ref{eq:gamma}, we have:
\begin{eqnarray} \gamma_{lk}=P(y=k|b=l)=P(y=k|v=i,b_{il}\neq 0
)\nonumber
\\=\frac{P(y=k)P(v=i,b_{il}\neq 0|y=k)}{\sum_{k=1}^K
P(y=k)P(v=i,b_{il}\neq 0|y=k)}\:\:\:\nonumber
\\=\frac{\exists_{b_{il}\neq 0}\Pi_{j=1}^n f(\theta_{kj},a_{ij})f(\delta_{kj},a_{ji})
\omega_k}{\sum_{k=1}^K \exists_{b_{il}\neq 0}\Pi_{j=1}^n
f(\theta_{kj},a_{ij})f(\delta_{kj},a_{ji}) \omega_k}
\label{eq:approximate_gamma}
\end{eqnarray}
\noindent where $\exists_{b_{il}\neq 0}$ denotes that randomly
selecting a node from block $l$.

That is, instead of averaging all nodes in the block $l$, the real
value of $\gamma_{lk}$ can be approximately estimated by a randomly
selected node from block $l$.

Correspondingly, an approximate version of the log-likelihood of
Eq.\ref{likelihood} is given by:

\begin{equation}
L(N|X,B_g)=\sum_{l=1}^L N_l(\exists_{b_{il}\neq 0}\ln \sum_{k=1}^K
\Pi_{j=1}^n
f(\theta_{kj},a_{ij})f(\delta_{kj},a_{ji})\omega_k)\label{appro_likelihood}
\end{equation}

\noindent where $N_l$ denotes the size of block $l$.

The time calculating Eqs.\ref{eq:gamma} and \ref{likelihood} will be
bounded by $O(n^2K)$. While, the time calculating
Eqs.\ref{eq:approximate_gamma} and \ref{appro_likelihood} will be
bounded by $O(LnK)$. This will be much efficient for constructing
the hierarchical organizations of networks.

\bigskip

\begin{Theo}
A local optimum of Eq.\ref{likelihood} will be guaranteed by
recursively calculating Eqs.4 and 5 in the article.
\end{Theo}

\begin{proof}
From the Proposition\ref{prop:1}, we have:

\noindent $L(N|X,B_g)$

\noindent$=\sum_{i=1}^n \ln P(v=i|X,B_g)$

\noindent$=\sum_{i=1}^n \ln \sum_{k=1}^K P(v=i,y=k|X,B_g)$

\noindent$=\sum_{i=1}^n \ln \sum_{k=1}^K
P(y=k|v=i,X^{(s)},B_g)\frac{P(v=i,y=k|X,B_g)}{P(y=k|v=i,X^{(s)},B_g)}$

\noindent(by Jensen's inequality)

\noindent$\geq\sum_{i=1}^n \sum_{k=1}^K
P(y=k|v=i,X^{(s)},B_g)\ln\frac{P(v=i,y=k|X,B_g)}{P(y=k|v=i,X^{(s)},B_g)}$

\noindent$\equiv G(X,X^{(s)})$.

Furthermore, we have:

\noindent$G(X^{(s)},X^{(s)})$

\noindent$=\sum_{i=1}^n\sum_{k=1}^K
P(y=k|v=i,X^{(s)},B_g)\ln\frac{P(v=i,y=k|X^{(s)},B_g)}{P(y=k|v=i,X^{(s)},B_g)}$

\noindent$=\sum_{i=1}^n\sum_{k=1}^K P(y=k|v=i,X^{(s)},B_g)\ln
P(v=i|X^{(s)},B_g)$

\noindent$=\sum_{i=1}^n \ln P(v=i|X^{(s)},B_g)\sum_{k=1}^K
P(y=k|v=i,X^{(s)},B_g)) $

\noindent$=\sum_{i=1}^n \ln P(v=i|X^{(s)},B_g)$

\noindent$= L(N|X^{(s)},B_g)$.

Let $P(y=k|b=l,X^{(s)},B_g)=\gamma_{ik}^{(s)}$, we have:

\noindent$G(X,X^{(s)})$

\noindent$=\sum_{l=1}^L\sum_{b_{il}\neq 0}\sum_{k=1}^K
\gamma_{lk}^{(s)}\ln P(v=i,y=k|X,B_g)-\sum_{l=1}^L\sum_{b_{il}\neq
0}\sum_{k=1}^K \gamma_{ik}^{(s)}\ln P(y=k|v=i,X^{(s)},B_g)$.

So, we have:

\noindent$\arg\max G(X,X^{(s)})$

\noindent$=\arg\max(\sum_{l=1}^L\sum_{b_{il}\neq 0}\sum_{k=1}^K
\gamma_{lk}^{(s)}\ln P(v=i,y=k|X,B_g)-\sum_{l=1}^L\sum_{b_{il}\neq
0}\sum_{k=1}^K \gamma_{ik}^{(s)}\ln P(y=k|v=i,X^{(s)},B_g))$

\noindent$=\arg\max(\sum_{l=1}^L\sum_{b_{il}\neq 0}\sum_{k=1}^K
(\gamma_{ik}^{(s)}\ln P(v=i,y=k|X,B_g)))$

\noindent$=\arg\max E[L(N,Z^{(s)}|X,B_g)]$

\noindent$=X^{(s+1)}$.

Recall that, the $\Theta^{(s+1)}$, $\Delta^{(s+1)}$ and
$\Omega^{(s+1)}$ of $X^{(s+1)}$ can be computed in terms of
$\gamma_{lk}^{(s)}$ by Eq.4 in the article. So, we have:

\centerline{$G(X^{(s+1)},X^{(s)})\geq
G(X^{(s)},X^{(s)})=L(N|X^{(s)},B_g)$.}

Recall that $L(N|X,B_g) \geq G(X,X^{(s)})$, we have:

\centerline{$L(N|X^{(s+1)},B_g)\geq G(X^{(s+1)},X^{(s)}) \geq
G(X^{(s)},X^{(s)})=L(N|X^{(s)},B_g)$.}

That is to say, the $X^{(s+1)}$ obtained in the current iteration
will be not worse than $X^{(s)}$ obtained in last iteration. So, we
have the theorem.
\end{proof}


\begin{Prop}
In terms of the parameter of $X$, $\Theta$, $\Delta$, $Z$ and
$\Omega$, we have:
\begin{equation}
\Phi=\Theta B_gZ D^{-1}, \quad \Psi=\Delta B_gZ D^{-1}
\end{equation} \noindent where $D=diag(n\Omega)$.
\end{Prop}

\begin{proof}

We have

\centerline{$\phi_{pq}=\sum_{i\in C_q}\frac{1}{N_q}\theta_{pi}$}

\noindent where $i\in C_q$ denotes node $i$ is in the cluster $q$
with a size $N_q$, and $\frac{1}{N_q}$ is the probability of
selecting node $i$ from cluster $q$. Furthermore, we have:

\centerline{$\phi_{pq}=\frac{1}{n\omega_q}\sum_{i=1}^n\theta_{pi}(B_gZ)_{iq}$.}

Similarly, we have:

\centerline{$\psi_{pq}=\frac{1}{n\omega_q}\sum_{i=1}^n\delta_{pi}(B_gZ)_{iq}$.}

So, we have

\centerline{$\Phi=\Theta B_gZ D^{-1}, \quad \Psi=\Delta B_gZ
D^{-1}$.}

\end{proof}

\begin{Algo} Searching the optimal model $X$ given $N$ and
$B_g$

$X$=GBM($N$,$B_g$)

01.\quad $K=1$;

02.\quad $X^{(0)}=LM(N,K,B_g)$; $L^{(0)}=-\ln(N|X^{(0)},B_g);$

03.\quad for $K$=2:$n$

04.\quad \quad \quad $X^{(1)}=LM(N,K,B_g)$;

05.\quad \quad \quad $L^{(1)}=-\ln(N|X^{(1)},B_g)+2gK^2\ln K^2;$

06.\quad \quad \quad if $L^{(1)}<L^{(0)}$ then

07.\quad \quad \quad \quad \quad $X^{(0)}= X^{(1)}$; $L^{(0)}=
L^{(1)}$;

08.\quad\quad \quad endif

09.\quad\quad \quad if $L^{(1)}$ keeps increasing for predefined
steps then

10.\quad\quad \quad \quad \quad return $X^{(0)}$;

11.\quad\quad \quad endif

12.\quad  endfor

13.\quad return $X^{(0)}$;

\end{Algo}

\bigskip

\begin{Algo} Searching the optimal model $X$ given $N$, $K$
and $B_g$

$X$=LM($N$,$K$,$B_g$)

01.\quad randomly initialing $\Gamma^{(0)}=(\gamma_{lk})_{L\times
K}$ with $\sum_k \gamma_{lk}=1$;

02.\quad $s\leftarrow 1$;

03.\quad repeat until convergence:

04.\quad\quad\quad compute $\Theta^{(s)}$, $\Delta^{(s)}$ and
$\Omega^{(s)}$;

05.\quad\quad\quad compute $\Gamma^{(s)}$;

06.\quad\quad\quad $s\leftarrow s + 1$;

07.\quad compute $Z$ from $\Gamma^{(s)}$ according to Bayesian rule;

\end{Algo}

\bigskip

\begin{Prop}
Let $B_{g_i}$ denotes the blocking model on the $i-th$ layer of the
hierarchical organization of network $N$, we have:

$P(X|N,B_{g_1},\cdots,B_{g_h})\propto P(N|X,B_{g_h})P(X)^{g_h}$
\end{Prop}

\begin{proof}

$P(X|N,B_{g_1},\cdots,B_{g_h})$

$=\frac{P(X,N,B_{g_1},\cdots,B_{g_h})}{P(N,B_{g_1},\cdots,B_{g_h})}$

$\propto P(X,N,B_{g_1},\cdots,B_{g_h})$

$=P(N|X,B_{g_1},\cdots,B_{g_h})P(X,B_{g_1},\cdots,B_{g_h})$

$=P(N|X,B_{g_1},\cdots,B_{g_h})P(X|B_{g_1},\cdots,B_{g_h})$

\quad $P(B_{g_h}|B_{g_1},\cdots,B_{g_{h-1}})\cdots
P(B_{g_2}|B_{g_1})P(B_{g_1})$

$\propto P(N|X,B_{g_1},\cdots,B_{g_h})P(X|B_{g_1},\cdots,B_{g_h})$

Since two nodes from the same block of $B_{g_{i-1}}$ will also be in
the same block of $B_{g_i}$, we have:

$P(X|N,B_{g_1},\cdots,B_{g_h})\propto
P(N|X,B_{g_h})P(X|B_{g_h})=P(N|X,B_{g_h})P(X)^{g_h}$

\end{proof}
\smallskip

\begin{Algo} Computing the threshold of a blocking model based on $\Phi$

01.\quad sort all entries of $\Phi$ into a non-increasing sequence
$S$;

02.\quad cluster all entries of $\Phi$ by the remarkable gaps of
$S$;

03.\quad return the biggest entry of the cluster with the minimum

\quad\quad mean as the threshold;
\end{Algo}

\smallskip
As examples, Fig. \ref{threshold} illustrates, by means of the above
algorithm, how to choose reasonable thresholds for reducing blocking
models of the world trade net and the co-purchased political book
network as discussed in the article.

\section{Additional Examples}

\subsection{Community detection from a benchmark network}

\noindent\smallskip We have applied our approach to the football
association network, a benchmark widely used for testing the
performance of community detection algorithms. Fig. \ref{fig:1}
gives the experimental results.

Our approach automatically find out 10 clusters or communities from
this network. Fig. \ref{fig:1}(a) shows the clustered network, and
Fig. \ref{fig:1}(c1) shows the clustered matrix, in which dots
denote the non-zero entries, and the rows and column corresponding
to the same cluster will be put together and be separated by the
solid lines. Fig. \ref{fig:1}b shows the searching process, in which
the cost in terms of $-L(N|X,B_1)+L_{\widehat{OC}}(X)$ firstly drops
down quickly, and then goes up sharply after a local minimum
corresponding to $K=10$. This iterative searching process shapes a
well-like landscape by penalizing both small likelihood and big
coding complexity or smaller priori, and a good compromise between
them is what we expect. Notice that, this estimated community number
is a little bit smaller than the real one, in total 12 real
associations, among which the teams from two small independent
associations prefer to play matches with outside teams. (c)-(e) show
the nested blocking models of this network. A hierarchical
organization with three layers have been constructed.

\subsection{Testing against synthetic networks}
We have tested our approach against several synthetic networks, and
here we give two examples. Fig. \ref{fig:2}a give one synthetic
network, in which a 2-community pattern and two bipartite patterns
coexist together. A two-layer nested organization are found out. In
the first layer, six clusters with homogeneous row as well as column
connection distributions are detected; in the second layer, such
detected clusters are grouped in pairs according to the node
couplings with coarser granularity. This time, a 2-community
pattern, a reciprocal bipartite pattern and an unilateral bipartite
pattern are emerged. Fig. \ref{fig:2}d give another synthetic
network, in which 12 communities are organized in a 3-layer
hierarchical structure according to the density of connections.
Correspondingly, a three-layer nested blocking model is constructed
as given in  Fig. \ref{fig:2}e.

\newpage
\noindent\textbf{Figure 1. An example of multiplex structures
consisting of hubs, outliers, and communities co-existing in one
social network.} The network as shown in (a) depicts the
co-appearances of 77 characters in the novel {\em Les Miserables}
(by Victor Hugo). Nodes denote characters and links connect any pair
of characters that appear in the same chapter of the novel. This
data set is from {\em The Stanford GraphBase: A Platform for
Combinatorial Computing}, edited by D.E.Knuth \cite{Knuth93}. The
physical connection profiles of nodes are represented in terms of
the adjacency matrix as shown in (b), in which dots denote ``1''
entries. In total, six blocks are detected by our method in terms of
the connection profiles of nodes, as separated by solid lines, so
that the nodes within the same blocks will demonstrate homogeneous
connectivities. To clearly visualize the block organization, the
matrix has been transformed by putting together the nodes within the
same blocks which are labeled by ``block 1'' to ``block 6'' from the
top down. Correspondingly, in the network as shown in (a), nodes are
also colored according to their block IDs (specifically, the
coloring schema is: block 1-red, block 2-green, block 3-blue, block
4-cyan, block 5-gray and block 6-yellow), and the sizes of nodes are
proportional to their respective degrees. (c) shows the blocking
model of the network, in which each block is colored, numbered and
placed according to the nodes it contains, and the sizes of blocks
are proportional to the number of nodes they contain, respectively.
The weights on the arrow lines globally measure the probabilities
that one node from one block will connect to another from other
blocks. (d) shows the reduced blocking model by cutting the arrow
lines with trivial weights. In this case, one hub pattern consisting
of block 1, one outlier pattern consisting of block 5 and four
community patterns consisting of blocks 2,3,4 and 6 will readily be
recognized by referring to the probabilistic linkage among these
blocks.

\noindent\textbf{Figure 2. The schematic illustrations of multiplex
structures.} (a)-(g) shows seven structural patterns frequently
observed in real-world networks, which are represented by networks,
matrices and blocking models, respectively. In the matrices, shades
represent the densities of links. In the blocking models, circles
denote blocks and solid arrow lines denote block
feedforward-couplings. From (a)-(g), structural patterns are
communities, authorities and outliers, hubs and outliers, a bow tie,
a multipartite, a bipartite and a hierarchical organization,
respectively. (g) shows a two-layer hierarchy; two overlapped
communities (also can be seen as a hub with two communities) and one
bipartite are in the first layer; two communities are in the second
layer.

\noindent\textbf{Figure 3. Detecting the nested core-periphery
organization and the cash flow patterns from a world trade net.} The
network as shown in (a) encodes the trade relationship among 80
countries. Nodes denote countries and arcs denote countries imported
commodities from others. The physical connection profiles of nodes
are represented in terms of its adjacency matrix as shown in (b), in
which dots denote ``1'' entries. In total, six blocks are detected
as separated by solid lines so that the nodes within the same blocks
will demonstrate homogeneous row- and column- connection
distributions. As before, the matrix has been transformed by putting
together the nodes within the same blocks which labeled by ``block
1'' to ``block 6'' from the top down. Correspondingly, in the
network shown in (a), nodes are colored according to their block IDs
(specifically, block 1-green, block 2-yellow, block 3-cyan, block
4-red, block 5-gray and block 6-blue), and the sizes of nodes are
proportional to their respective out-degrees. (c) shows the world
map in which 80 countries are colored by the same coloring schema
defined above. (d) shows the detected two-layer hierarchical
organization, in which each block is colored by the same coloring
schema defined above and the sizes of blocks are proportional to the
number of nodes they contain, respectively. In the first layer, a
reduced blocking model is given by cutting the trivial block
couplings. The arrow lines with weights denote the cash flows from
the countries within one block to others. The cash flows from two
hubs (blocks 4 and 1) are highlighted by thicker arrow lines, of
which the thickness is proportional to their strength measured by
respective block couplings. These highlighted cash flows would
outline the backbone of the cash circulation in the world trade
system. A macroscopic hub-outlier pattern on the second layer is
detected; together with the microscopic hub-outlier patterns on the
first layer, the whole trade system would demonstrate a nested
core-periphery organization. (e) shows the multiplex structures
discovered from the reduced blocking model on the first layer by a
procedure of matching isomorphism subgraphs as defined before. These
structural patterns are stored in their respective reservoirs
labeled as authority, community, hub or outlier.

\noindent\textbf{Figure 4. Mining granular association rules from a
co-occurrence net.} The network in (a) encodes the co-purchased
relationship among 105 books, in which the nodes with large, median
and small sizes are labeled by ``liberal'', ``conservative'' and
``neutral(including one unlabeled)'', respectively. A two-layer
hierarchical organization is detected. In the first layer, seven
blocks are detected. As before, the matrix as shown in (b) has been
transformed by putting together the nodes within the same blocks
which labeled by ``block 1'' to ``block 7'' from the top down.
Correspondingly, in the network shown in (a), nodes are colored
according to their block IDs (specifically, block 1-blue, block
2-red, block 3-gray, block 4-brown, block 5-green, block 6-cyan and
block 7-yellow). In the blocking model as shown in (b), each block
is colored by the same coloring schema defined above, the sizes of
blocks are proportional to the number of nodes they contain,
respectively; the arrow lines to be reserved in its reduced blocking
model are highlighted by thicker arrow lines. The blocking model and
its corresponding matrix of the second layer are shows in (c), in
which a global two-community structure is detected. The granularity
of the network corresponding to different layers is given by $g$.

\newpage

\begin{figure}[b]
 \centering
\centerline{\hbox{\psfig{figure=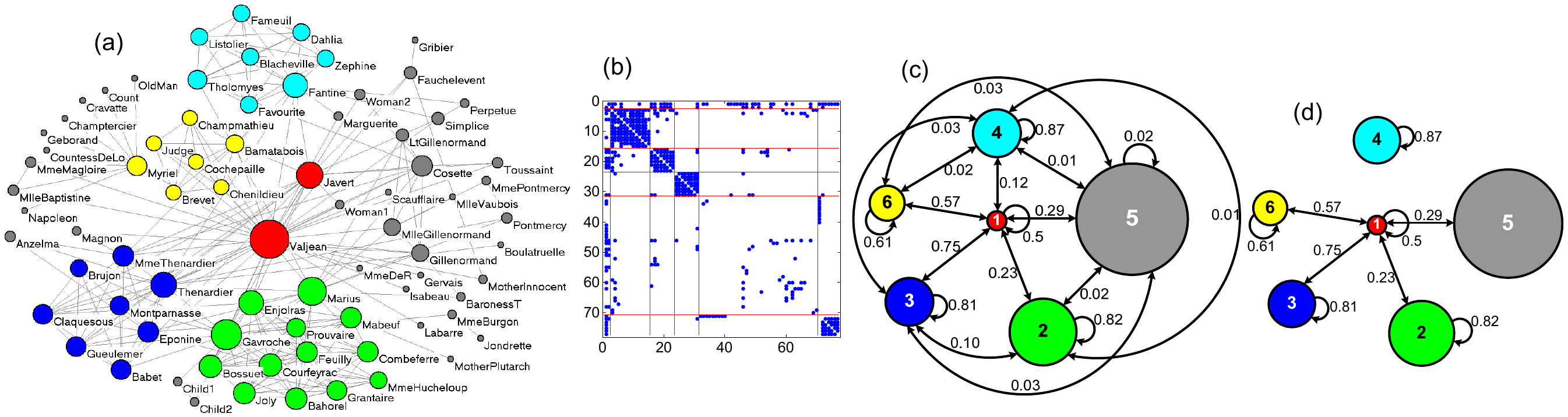,scale=0.7}}}
\caption{An example of the multiplex structures in a social
network.} \label{lesmis}
\end{figure}

\newpage

\begin{figure}[b]
  \centering
\centerline{\hbox{\psfig{figure=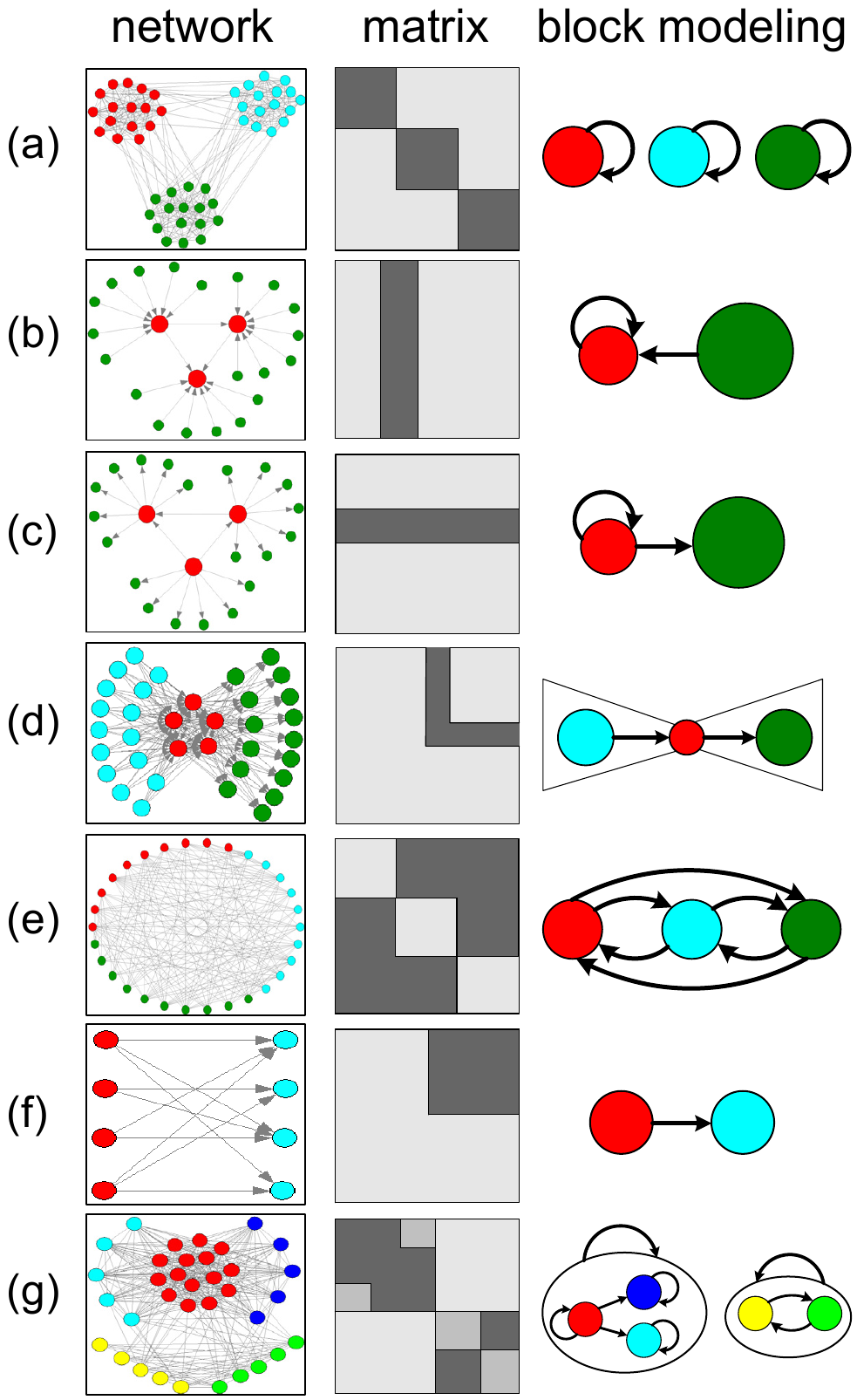,scale=0.7}}}
  \caption{The schematic illustrations of multiplex
structures.} \label{blocking}
\end{figure}

\newpage

\begin{figure}[b]
  \centering
\centerline{\hbox{\psfig{figure=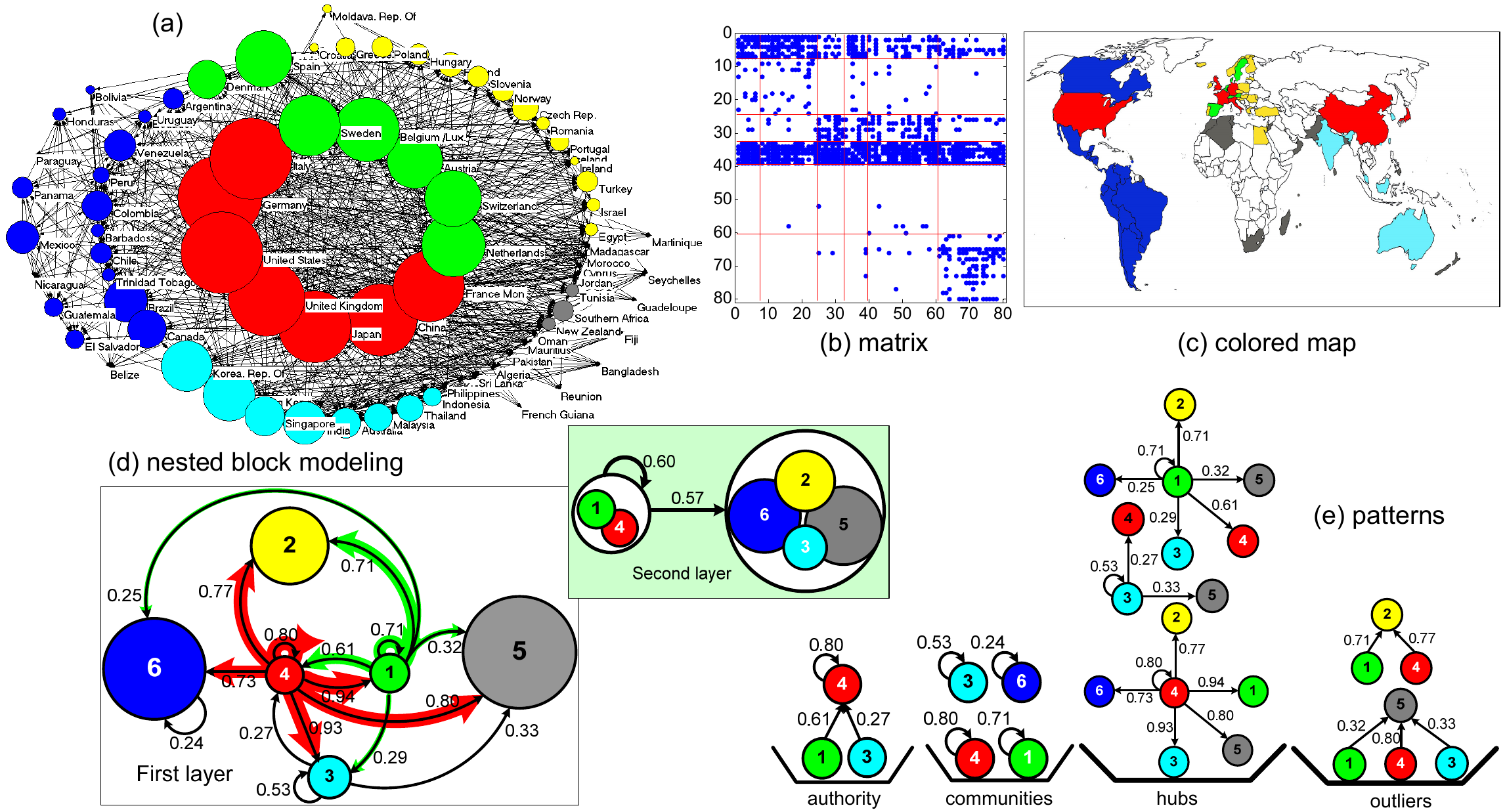,scale=0.7}}}
  \caption{Detecting the nested core-periphery organization and the
cash flow patterns from a world trade net.} \label{worldtrade}
\end{figure}

\newpage

\begin{figure}[b]
  \centering
\centerline{\hbox{\psfig{figure=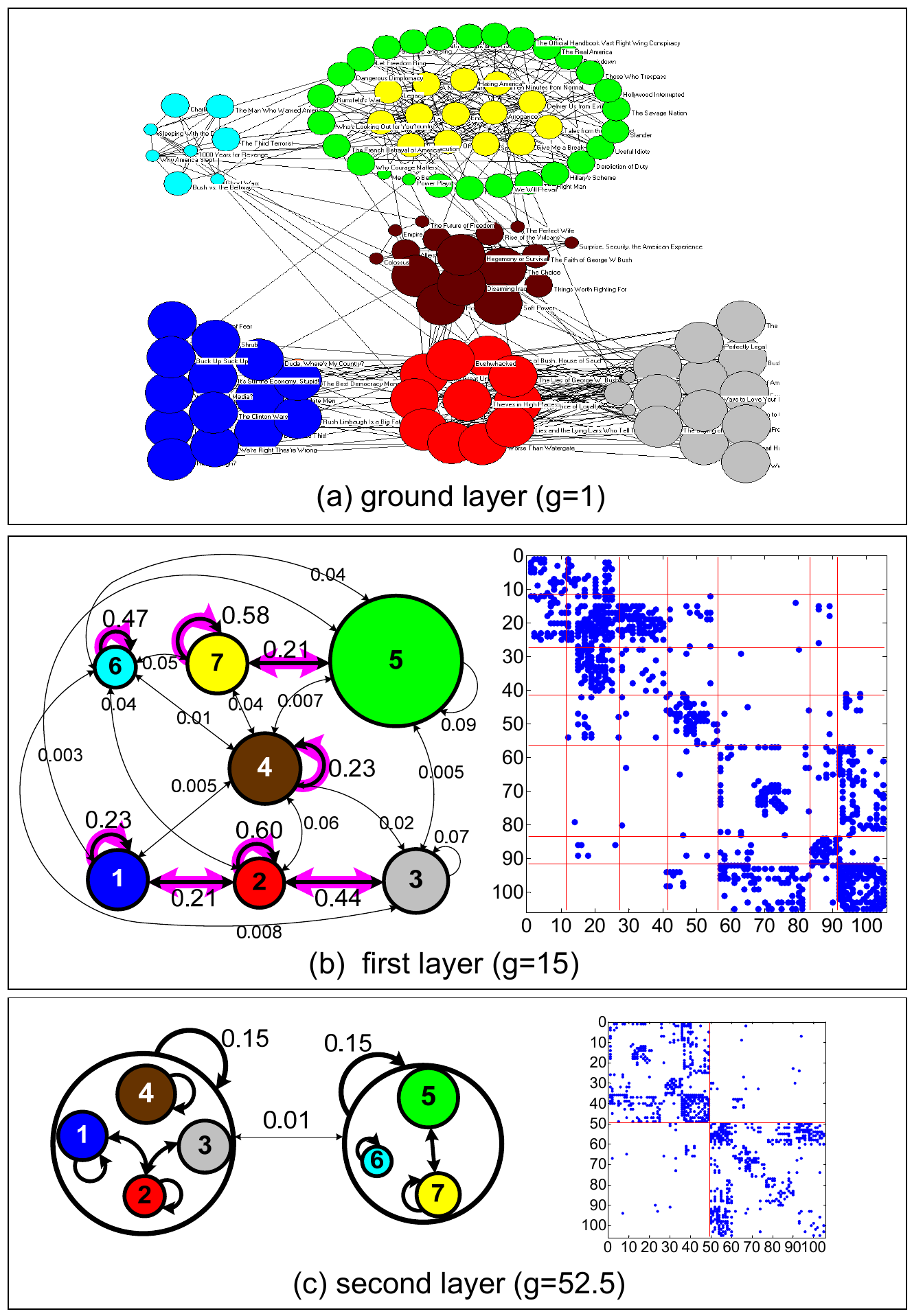,scale=0.75}}}
  \caption{Mining granular association rules from a
co-occurrence net.}\label{polbooks}
\end{figure}

\begin{figure}[b]
  \centering
\centerline{\hbox{\psfig{figure=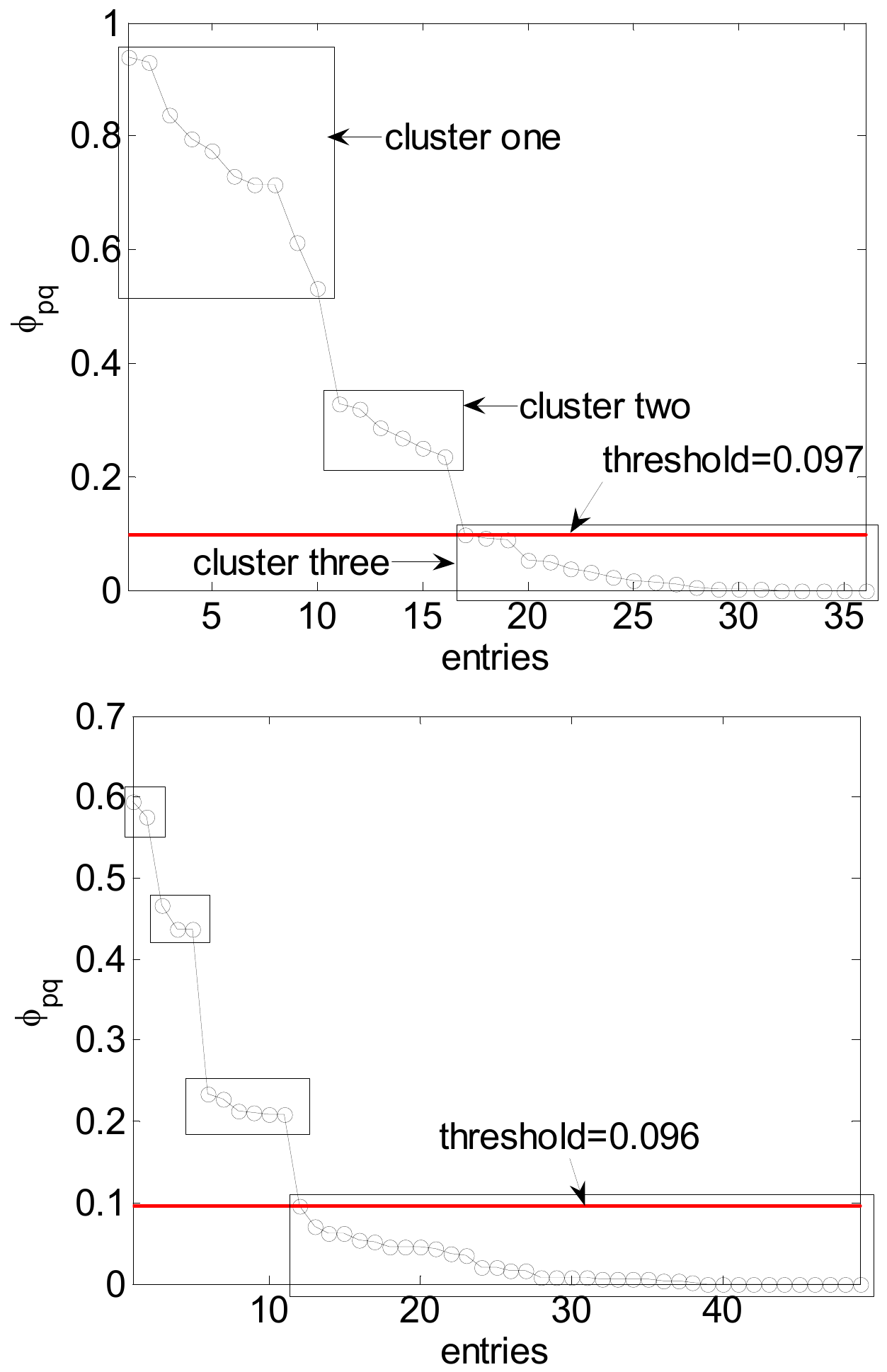,scale=0.65}}}
  \caption{The illustrations of calculating thresholds for reducing blocking models.
  The top: calculating the threshold for the blocking model of the first layer of the world trade
  net. The bottom: calculating the threshold for the blocking model of the first layer of the
  co-purchased political book network. In both figures, the $x$ denotes the sorted entries and the $y$ denotes the values of block couplings in terms of $\Phi$.
  These couplings are clustered into three or four groups, as separated by rectangles, by remarkable
  gaps. The calculated thresholds are the maximum entries of the clusters with minimum means pointed out by red solid lines, respectively. \label{threshold}}
\end{figure}

\begin{figure}[b]
  \centering
\centerline{\hbox{\psfig{figure=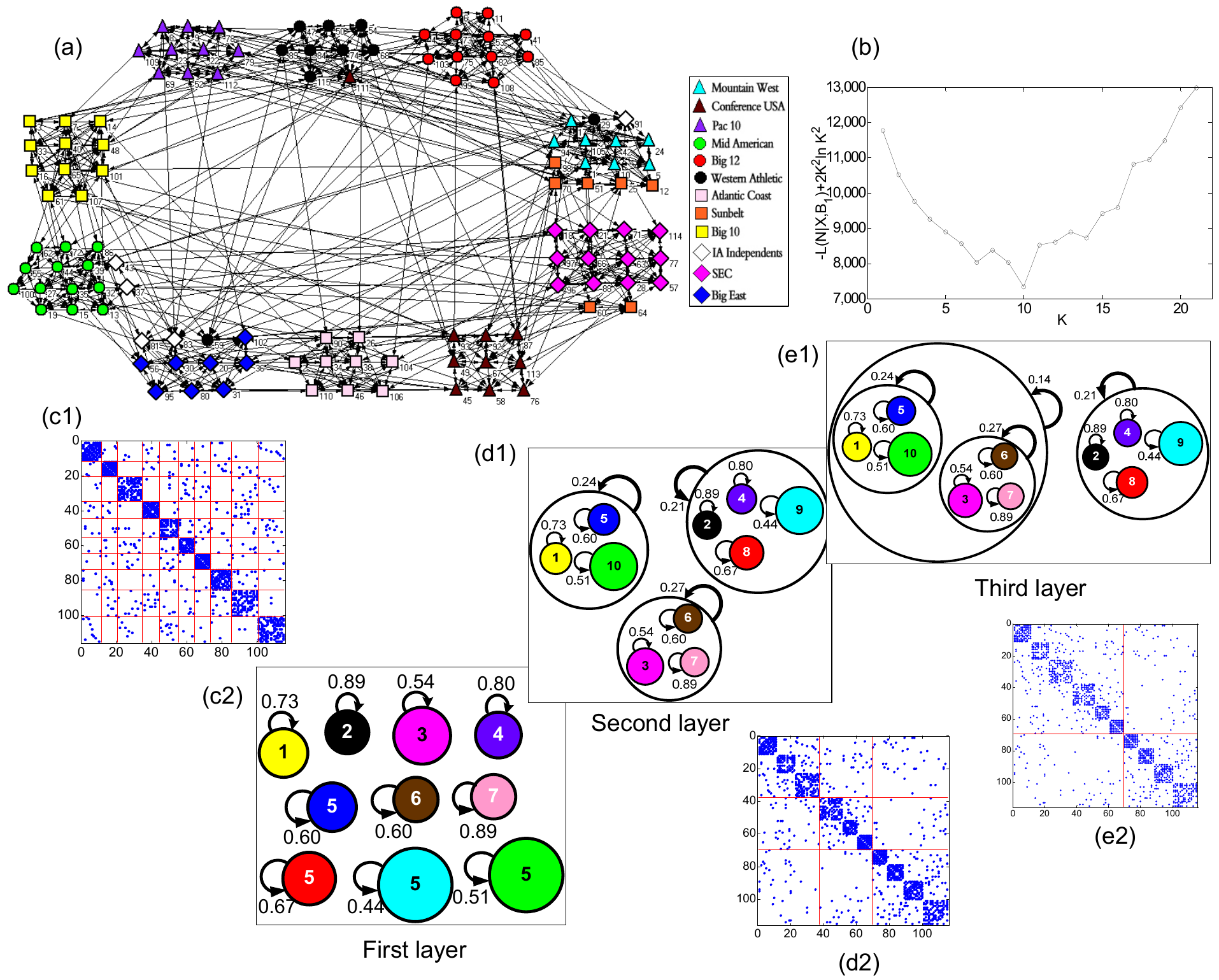,scale=0.75}}}
  \caption{Detecting communities from the football association network. (a) The clustered football association network;
  (b) The iterative searching landscape. (c1-c2) The matrix and blocking model of the first layer.
  (d1-d2) The matrix and blocking model of the second layer.
  (e1-e2) The matrix and blocking model of the third layer.
  The coloring schemes used in (a) and (c2), (d1) and (e1) are same.\label{fig:1}}
\end{figure}

\begin{figure}[b]
  \centering
\centerline{\hbox{\psfig{figure=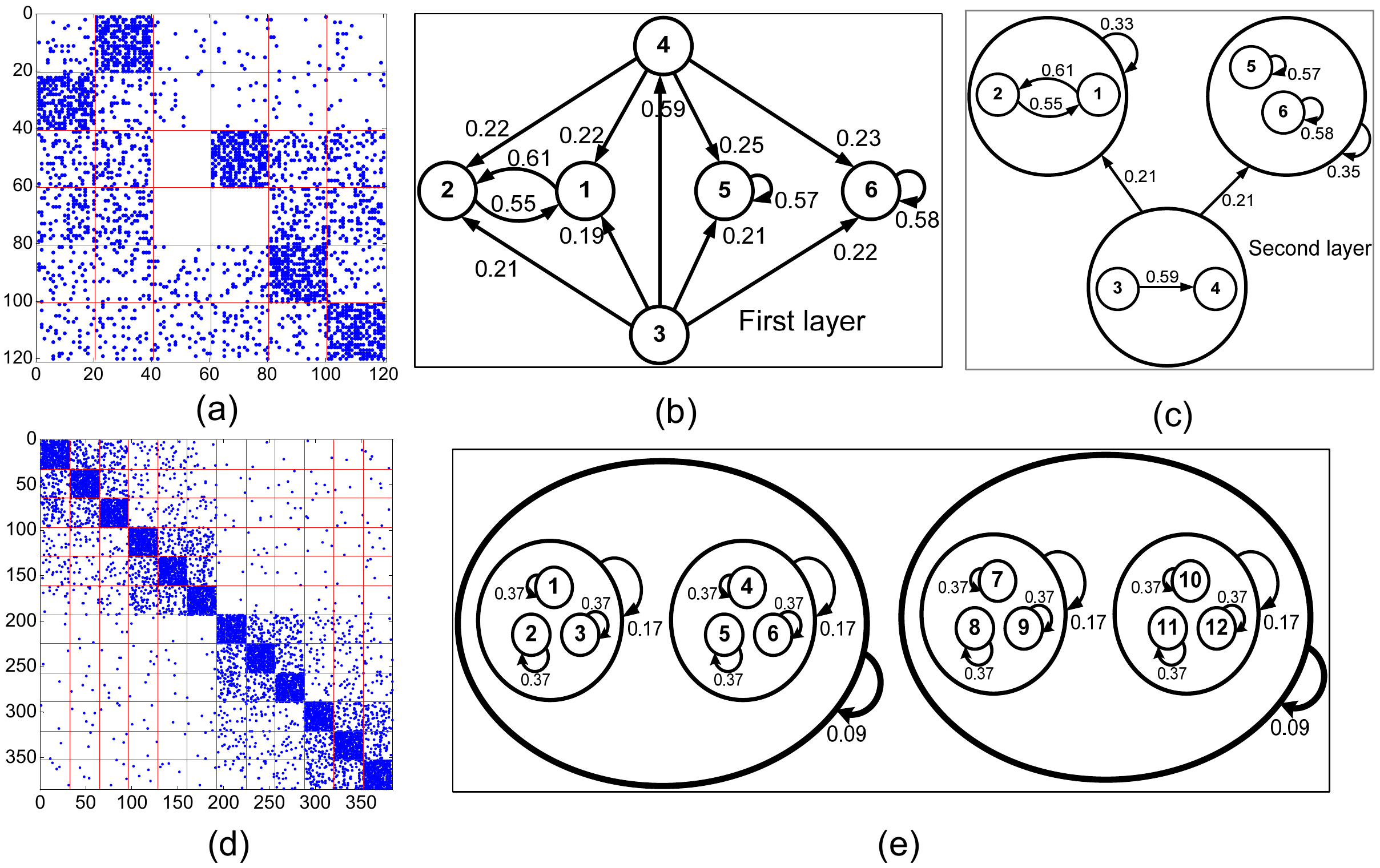,scale=0.55}}}
  \caption{Testing synthetic networks. (a) The matrix of a synthetic network, in which
  dots denote the non-zero entries, and the six detected clusters  are separated by solid lines;
  (b) The blocking model in the first layer, in which blocks are numbered according to the same sequence of the clusters in the matrix.
  (c) The blocking model in the second layer. (d) Another synthetic network, in which twelve detected clusters are separated
  by the solid lines. (e) The obtained nested blocking model corresponding to the hierarchical organization with three layers,
  in which blocks are numbered according to the same sequence of the clusters in the matrix. \label{fig:2}}
\end{figure}


\begin{thebibliography}{10}
\bibitem{report2006}Boccaletti S, Latora V, Moreno Y, Chavez M, Hwang DU(2006) Complex networks: Structure and dynamics.
\emph{Physics Reprots} 424:175-308.


\bibitem{Milo2002} Milo R, Orr SS, Itzkovitz S,
Kashtan N, Chklovskii D, Alon U(2002) Network Motifs: simple
building blocks of complex networks. \emph{Science} 298:824-827.

\bibitem{Newman2002}Girvan M, Newman MEJ(2002) Community structure in social and biological networks. \emph{Proc Natl Acad Sci USA} 9:7821-7826.

\bibitem{Watts1998} Watts DJ, Strogatz SH(1998) Collective Dynamics of Small-World Networks. \emph{Nature} 393:440-442.

\bibitem{Barabasi1999} Barabasi AL, Albert R(1999) Emergence of Scaling in Random Networks. \emph{Science} 286:509-512.


\bibitem{report2010}Fortunato S(2010) Community detection in graphs. \emph{Physics Reports} 486:75-174.


\bibitem{bipartite2003}Holme P, Liljeros F, Edling CR, Kim BJ(2003) Network bipartivity. \emph{Phys Rev E}
68:056107.

\bibitem{bipartite2004} Guillaume JL, Latapy M(2004) Bipartite structure of all complex networks. \emph{Inform Process Lett} 90:215-221.

\bibitem{bipartite2009} Brady A, Maxwell K, Daniels N, Cowen LJ(2009) Fault Tolerance in Protein Interaction Networks: Stable Bipartite Subgraphs and Redundant Pathways. \emph{PLoS ONE} 4:e5364.

\bibitem{K1999} Kleinberg JM(1999) Authoritative Sources in a Hyperlinked Environment. {\it Journal of ACM} 46:604-632.

\bibitem{Albert2000} Albert R, Jeong H, Barabasi AL(2000) The Internet's Achilles heel: Error and attack tolerance of complex netowrks. \emph{Nature}
406:378-382.

\bibitem{Sporns2007}Sporns O, Honey C, Kotter R(2007) Identification and classification of hubs in brain networks. \emph{PLoS ONE} 2:e1049.

\bibitem{Bowtie1999}Broder A, Kumar R, Maghoul F, Raghavan P, Rajagopalan S, Stata R,
Tomkins A, Wiener J(1999) Graph structure in the Web. \emph{COMPUT
NETW} 33:309-320.


\bibitem{bowtie2000} News Feature(2000) The web is a bow tie. \emph{Nature} 405:113.

\bibitem{bowtie2003}Ma HW, Zeng AP(2003) The connectivity structure, giant strong component and centrality of metabolic networks. \emph{Bioinformatics} 19:1423-1430.

\bibitem{Palla2005}Palla G, Derenyi I, Farkas I, Vicsek T(2005) Uncovering the overlapping community structures of complex networks in nature and society. \emph{Nature} 435:814-818.


\bibitem{Knuth93}Knuth DE(1993). \emph{The Stanford GraphBase: A Platform for Combinatorial
Computing} (Addison-Wesley press, Reading, MA).


\bibitem{hierarchyscience2004} Ravasz E, Somera AL, Mongru DA, Oltvai ZN, Barabasi
AL(2002) Hierarchical organization of modularity in metabolic
networks. \emph{Science} 297:1551-1555.

\bibitem{hierarchy2006}Zhou C, Zemanova L, Zamora G, Hilgetag CC, Kurths J(2006) Hierarchical Organization Unveiled by Functional Connectivity in Complex Brain Networks. \emph{Phys Rev Lett} 97:238103.


\bibitem{hierarchy2007}Pardo MS, Guimera R, Moreira AA, Amaral LAN(2007) Extracting the hierarchical organization of complex systems. \emph{Proc Natl Acad Sci USA}
104:7821-7826.

\bibitem{hierarchy2008} Clauset A, Moore C, Newman MEJ(2008) Hierarchical structure and the prediction of missing links in networks. {\it Nature} 453:98-101.

\bibitem{multiplexform2008} kemp C, Tenenbaum JB(2008) The discovery of structural form. \emph{Proc Natl Acad Sci USA} 105:10687-10692.

\bibitem{block1971} Lorrain F, White HC(1971) Structural equivalence of individuals in social networks. \emph{J MATH SOCIOL}
1:49-80.

\bibitem{block1983a} White DR, Reitz KP(1983) Graph and semigroup homomorphism on networks of relations. \emph{Social
Networks} 5:193-235.

\bibitem{block1981}Fienberg SE, Wasserman S(1981) Categorical data analysis of single sociometric relations. \emph{Sociological methodology} 12:156-192.

\bibitem{block1983b}Holland PW, Laskey KB, Leinhardt S(1983) Stochastic blockmodels: Some first steps. \emph{Social Networks} 5:109-137.

\bibitem{pans2007} Newman MEJ, Leicht EA(2007) Mixture models and exploratory analysis in networks. \emph{Proc Natl Acad Sci USA} 104:9564-9569.

\bibitem{Dempster1977} Dempster AP, Laird NM, Rubin DB(1977) Maximum likelihood from incomplete data via the EM algorithm. \emph{J R
Statist Soc B} 39:185-197.

\bibitem{shannon1949} Shannon CE, Weaver W(1949) \emph{The mathematical
theory of communication} (University of Illinois Press, Urbana).

\bibitem{nooy} Nooy WD, Mirvar A, Batagelj V(2004) \emph{Exploratory social
network analysis with Pajeck} (Cambridge University Press)

\bibitem{GDP}Smith DA, White DR(1992) Structure and Dynamics of the Global
Economy - Network Analysis of International-Trade 1965-1980.
\emph{Social Forces} 70:857-893.

\bibitem{polbook}Newman MEJ(2006) Finding community structure in networks using the eigenvectors of matrices. \emph{Phys Rev E} 74:036104.

\end{thebibliography}
\end{document}